\documentclass[12pt]{article}
\usepackage{amssymb,amsthm,amsmath,fullpage}
\usepackage{arydshln}
\usepackage{tikz}
\usetikzlibrary{matrix}
\usepackage{graphics}
\usepackage{booktabs,verbatim} 

\newtheorem{theorem}{Theorem}
\newtheorem{lemma}[theorem]{Lemma}
\newtheorem{proposition}[theorem]{Proposition}

\theoremstyle{definition}
\newtheorem*{definition}{Definition}
\newtheorem*{remark}{Remark}
\newtheorem{example}[theorem]{Example}

\newenvironment{keywords}{{\bf Keywords:}}{}

\newcommand{\ff}{{\mathbb{F}}}

\title{Isometry-Dual Flags of AG Codes}
\author{Maria Bras-Amor\'os,  Iwan Duursma and Euijin Hong}
\begin{document}
\maketitle

\begin{abstract}
  Consider a complete flag $\{0\} = C_0 < C_1 < \cdots < C_n =  \mathbb{F}^n$ of one-point AG codes of length $n$ over the finite field $\mathbb{F}$. The codes are defined by evaluating functions with poles at a given point $Q$ in points $P_1,\dots,P_n$ distinct from $Q$. A flag has the isometry-dual property if the given flag and the corresponding dual flag are the same up to isometry. For several curves, including the projective line, Hermitian curves, Suzuki curves, Ree curves, and the Klein curve over the field of eight elements, the maximal flag, obtained by evaluation in all rational points different from the point $Q$, is self-dual. More generally, we ask whether a flag obtained by evaluation in a proper subset of rational points is isometry-dual. In \cite{GMRT} it is shown, for a curve of genus $g$, that a flag of one-point AG codes defined with a subset of $n > 2g+2$ rational points is isometry-dual if and only if the last code $C_n$ in the flag is defined with functions of pole order at most $n+2g-1$. Using a different approach, we extend this characterization to all subsets of size $n \geq 2g+2$. Moreover we show that this is best possible by giving examples of isometry-dual flags with $n=2g+1$ such that $C_n$ is generated by functions of pole order at most $n+2g-2$. We also prove a necessary condition, formulated in terms of maximum sparse ideals of the Weierstrass semigroup of $Q$, under which a flag of punctured one-point AG codes inherits the isometry-dual property from the original unpunctured flag.
\end{abstract}


\section{Introduction}\label{s:Introduction}

Let $\ff$ be a finite field. A linear code $C$ of length $n$ over $\ff$ is a subspace $C \subset \ff^n$. Using the standard inner product, the space $\ff^n$ is an inner product space. The dual code $C^\perp$ of a code $C$ is defined as the orthogonal complement of $C$ in the inner product space $\ff^n$.
A complete flag of codes is an increasing sequence $\{0\} = C_0 < C_1
< \cdots < C_n = \mathbb{F}^n$ of linear subspaces of $\mathbb{F}^n$. The dual of the flag is the flag $\{0\} = C_n^\perp < C_{n-1}^\perp < \cdots < C_0^\perp = \mathbb{F}^n$. In this paper we consider duality properties of flags of codes, and in particular flags of one-point AG codes. 

A flag is self-dual if it is identical to its dual flag. In general we
use a weaker notion of duality and we say that a flag is isometry-dual
if the dual of the flag is the image of the original flag under an
isometry. A linear isometry for $\ff^n$ is an invertible linear map
${\mathbb F}^n\rightarrow{\mathbb F}^n$ preserving the Hamming distance $d_H$ , which is defined as the number of positions in which two vectors have different coordinates.  

A linear isometry can be written uniquely as the composition of a permutation matrix and an invertible diagonal matrix. AG codes are defined with an ordered list $P_1, P_2, \ldots P_n$ of $n$ rational points. For isometries between AG codes we assume that the codes are defined with the same ordered list of $n$ points and that isometries preserve this order, in other words that isometries do not permute coordinates.
With this in mind, we say that a complete flag of AG codes is isometry-dual if the dual flag is the image of the original flag under right multiplication by an invertible diagonal matrix.



Let $\{0\} = C_0 < C_1 < \cdots < C_n = \mathbb{F}^n$ be an isometry-dual flag of linear codes with isometry
matrix $M$, such that each $C_i$ is the row space of a matrix $G_i$,
for $i=0,1,\ldots,n.$ Then the dual codes are the row spaces of
matrices $G_i M$, for $i=0,1,\ldots,n,$ so that, 
for $i=0,1,\ldots,n$,
\[
G_i (G_{n-i} M)^T = G_i M^T G_{n-i}^T = 0. 
\]
For the codes in a complete flag, the matrices $G_i$ can be chosen such that each $G_i$ is the leading $i \times n$ submatrix of an $n \times n$ matrix $G$. For such a choice, the flag is isometry-dual if there exists an isometry with matrix $M$ such that
\begin{equation} \label{eq:gmg}
( G M^T G^T )_{i,j} = 0, ~~~\text{for $i+j \leq n$}, \qquad  ( G M^T G^T )_{i,j} \neq 0, ~~~\text{for $i+j = n+1$}.
\end{equation}
We give two examples of self-dual flags. 
For a flag of Reed-Solomon type over a field $\mathbb{F}_q$ of size $q$ we evaluate the monomials  $1, x, \ldots, x^{q-1}$ in the elements $\alpha_1, \ldots, \alpha_q$ of the field $\ff_q$. For $q=8$ the matrices $G$ and $G G^T$ are
 \[
{\small \begin{array}{ccccccccccccc}
\toprule 
& &1 &\alpha&\alpha^2 &\alpha^3 &\alpha^4 &\alpha^5 &\alpha^6 &0 \\
\midrule
1 & &1 &1 &1 &1 &1 &1 &1  &1 \\
x & &1 &\alpha &\alpha^2 &\alpha^3 &\alpha^4 &\alpha^5 &\alpha^6 &0 \\
x^2 & &1 &\alpha^2 &\alpha^4 &\alpha^6 &\alpha &\alpha^3 &\alpha^5 &0 \\
x^3 & &1 &\alpha^3 &\alpha^6 &\alpha^2 &\alpha^5 &\alpha &\alpha^4 &0 \\
x^4 & &1 &\alpha^4 &\alpha &\alpha^5 &\alpha^2 &\alpha^6 &\alpha^3 &0 \\
x^5 & &1 &\alpha^5 &\alpha^3 &\alpha &\alpha^6 &\alpha^4 &\alpha^2 &0 \\
x^6 & &1 &\alpha^6 &\alpha^5 &\alpha^4 &\alpha^3 &\alpha^2 &\alpha &0 \\
x^7 & &1 &1 &1 &1 &1 &1 &1 &0 \\
\bottomrule
\end{array} \qquad 
G G^T = 
\left[ \begin{array}{cccccccc}
0 &0 &0 &0 &0 &0 &0 &1 \\
0 &0 &0 &0 &0 &0 &1 &0 \\
0 &0 &0 &0 &0 &1 &0 &0 \\
0 &0 &0 &0 &1 &0 &0 &0 \\
0 &0 &0 &1 &0 &0 &0 &0 \\
0 &0 &1 &0 &0 &0 &0 &0 \\
0 &1 &0 &0 &0 &0 &0 &0 \\
1 &0 &0 &0 &0 &0 &0 &0 \\
\end{array}
\right]}
\]
For a flag of Hermitian type we choose a curve $y^q+y = x^{q+1}$ over a field $\ff$ of size $q^2$ and evaluate monomials $x^i y^j$ in the $q^3$ points $(a,b) \in \ff \times \ff$ that satisfy $b^q+b = a^{q+1}$. Monomials $x^i y^j$, for $0 \leq  i \leq q^2-1$, $0 \leq j \leq q-1$, are ordered according to their weighted degree $qi+(q+1)j$. The monomials have a single pole at infinity and the weighted degree agrees with the pole order. For $q=2$ and 
$\ff = \{ 0,1,\omega,{\bar \omega}\}$ the matrices $G$ and $G G^T$ are
 \[
{\small \begin{array}{ccccccccccccc}
\toprule 
& &(1,\omega) &(1,{\bar \omega}) &(\omega,\omega) &(\omega,{\bar \omega}) &({\bar \omega},\omega) &({\bar \omega},{\bar \omega}) &(0,0) &(0,1) \\
\midrule
1 & &1 &1 &1 &1 &1 &1 &1  &1 \\
x & &1 &1 &\omega &\omega &{\bar \omega} &{\bar \omega} &0 &0 \\
y & &\omega &{\bar \omega} &\omega &{\bar \omega} &\omega &{\bar \omega} &0 &1 \\
x^2 & &1 &1 &{\bar \omega} &{\bar \omega} &\omega &\omega &0 &0 \\
xy & &\omega &{\bar \omega} &{\bar \omega} &1 &1 &\omega &0 &0 \\
x^3 & &1 &1 &1 &1 &1 &1 &0 &0 \\
x^2y & &\omega &{\bar \omega} &1 &\omega &{\bar \omega} &1 &0 &0 \\
x^3y & &\omega &{\bar \omega} &\omega &{\bar \omega} &\omega &{\bar \omega} &0 &0 \\
\bottomrule
\end{array} \qquad  
G G^T = 
\left[ \begin{array}{cccccccc}
0 &0 &0 &0 &0 &0 &0 &1 \\
0 &0 &0 &0 &0 &0 &1 &0 \\
0 &0 &0 &0 &0 &1 &0 &1 \\
0 &0 &0 &0 &1 &0 &0 &0 \\
0 &0 &0 &1 &0 &0 &1 &0 \\
0 &0 &1 &0 &0 &0 &0 &1 \\
0 &1 &0 &0 &1 &0 &0 &0 \\
1 &0 &1 &0 &0 &1 &0 &1 \\
\end{array}
\right]}
\]
In both cases the flag is self-dual and (\ref{eq:gmg}) holds with $M=I$. The main question in this paper is to decide whether the projection of a flag on a subset of coordinates preserves the isometry-dual property. The projection of a flag is described by a square submatrix of the original matrix $G$. To obtain the  submatrix for a projection on $k$ coordinates we first choose a submatrix of $k$ columns. In the submatrix we then select $k$ linearly independent rows, greedily from top to bottom. As an example we select the columns $(\omega,\omega), (\bar \omega, \bar \omega), (0,1).$ The top three rows in the submatrix are linearly independent. The reduced flag of codes has matrix $G$ of size $3$. The matrix satisfies (\ref{eq:gmg}), with $M \neq I,$ and the reduced flag is isometry-dual.
\[
\begin{array}{ccccccccccccc}
\toprule 
& &(\omega,\omega)  &({\bar \omega},{\bar \omega})  &(0,1) \\
\midrule
1 & &1 &1 &1 \\
x & &\omega  &{\bar \omega} &0 \\
y &  &\omega &{\bar \omega} &1 \\
\bottomrule
\end{array} \qquad  
G  
\left[ \begin{array}{cccccccc}
{\bar \omega} &0 &0 \\
0 &\omega &0 \\
0 &0 &1 \\
\end{array}
\right] G^T =
\left[ \begin{array}{cccccccc}
0 &0 &1 \\
0 &1 &1 \\
1 &1 &0 \\
\end{array}
\right].
\]
The inheritance of the isometry-dual property is different for Reed-Solomon codes and for Hermitian codes. When a Reed-Solomon code is restricted to a subset of coordinates the dual code
becomes a generalized Reed-Solomon code \cite[Chapter~10, Theorem~4]{MWS}. Generalized Reed-Solomon codes are isometric to Reed-Solomon codes and thus any reduced flag is isometry-dual. On the other hand, for Hermitian codes only certain subsets of columns define isometry-dual flags. 


We first recall the definition of one-point AG codes. 
Throughout this paper, let $\mathcal{X}$ be a smooth absolutely irreducible pojective curve of genus $g$ defined over the finite field $\ff$ and let $P_1, \ldots, P_n$ and $Q$ be distinct rational points on $\mathcal{X}$. Define a divisor $D=P_1+P_2 + \ldots + P_n$ and let $G$ be a divisor with support disjoint from $D$. 
The AG code $C_L(D, G) \subset \ff^n$ is defined as the image of the evaluation map
\[
\text{ev}_D : L(G) \longrightarrow \mathbb{F}^n, \quad
\text{ev}_D(f) = (f(P_1), f(P_2), \ldots, f(P_n)).
\]
One-point AG codes are defined with divisor $G$ a multiple of $Q$. To define a complete flag of one-point AG codes, let $m_0 = -1$ and $C_0 = C_L(D,m_0Q) = 0$, and choose the remaining codes in the flag of the form $C_i= C_L(D, m_iQ)$, where the $m_i$ are mininal such that
$\dim C_L(D, m_iQ) = i,$ for $i=1, \ldots, n$.
In particular, $m_1 = 0$. Then the sequence $\{0\} = C_0 < C_1 < \cdots < C_n = \mathbb{F}^n$ is a complete flag of linear codes.
We call the $m_i$, for $i=1,2, \ldots, n,$ the \textit{geometric nongaps} corresponding to $D$ and $Q$ and write $W^* = \{ m_1, \ldots, m_n \}$. 
Denoting the semigroup of Weierstrass nongaps of $\mathcal{X}$ at $Q$ by $W$, we have $W^* \subset W$. In this paper, in the context of one-point AG codes, a {\em complete flag} is not only the flag as vector spaces but also constructed with geometric nongaps.

\begin{lemma}\label{lemmacr1cr2}
	An element $a \in W$ is a geometric nongap if and only if it satisfies the following two criterion.
	\begin{align}
	&L(aQ) \neq L((a-1)Q) \label{cr1}\\
	&L(aQ-D) = L((a-1)Q-D) \label{cr2}
	\end{align}
\end{lemma}

\begin{proof}
	Note that $C_L(D, aQ) \simeq L(aQ)/L(aQ-D)$. In both
        (in)equalities, the dimension of the left hand side is greater
        or equal to that of the right hand side at most 1. Then the
        first inequality is the criteria for a Weierstrass nongap. The
        second one is induced from the first one and the fact that
        $C_L(D, aQ) \gneq C_L(D, (a-1)Q)$ 
\end{proof}

For a complete flag of one-point AG codes $(C_L(D, m_i Q))_{i=0, \ldots, n}$, we call the pair $(n,m_n)$ \textit{admissible}. Let $m:=m_n$. Our main theorem restricts the admissible pairs $(n,m)$ to four cases.

\begin{theorem}\label{t:main}
	Let $\big(C_L(D, m_iQ)\big)_{i=0, \ldots, n}$ be a complete flag of one-point AG codes. In general, $m \leq n+2g-1$ and the following hold. 
	\begin{enumerate}
		\item[(a)] $n=1$ if and only if $m=0$.
		\item[(b)] $n \leq  \dfrac{m}2 +1 \leq g+1$ for $0 < m \leq 2g$.
	\end{enumerate}
	Moreover, for isometry-dual flags, 
	\begin{enumerate}
		\item[(c)] $n \leq  \dfrac{m}2 + \dfrac32 < 2g+2$ for $2g < m \leq 4g$.
		\item[(d)] $n = m+1-2g \geq 2g+2$ for $m > 4g$.
	\end{enumerate}
	Conversely, a complete flag with $n = m+1-2g \geq 2g+2$ is isometry-dual.
\end{theorem}

The cases (a) and (b) hold for general complete flags and are easily dealt with in Lemma~\ref{l:mn}. The cases (c) and (d) hold for isometry-dual flags and are proved in Proposition~\ref{p:main}. The converse property is proven in Proposition \ref{c:main}.

\begin{remark}
	According to Theorem \ref{t:main}, admissable pairs $(n,m)$ are of one of four types. They are 
	\begin{enumerate}
		\item[(a)] $(n,m)=(1,0)$.
		\item[(b)] $(n,m)$ is in the convex hull of $(2,2), (2,2g+1), (g+1,2g)$ and $(g+1,3g)$.
		\item[(c)] $(n,m)$ is in the convex hull of $(g+2,2g+1), (g+2,3g+1), (2g+1,4g-1)$ and $(2g+1,4g)$.
		\item[(d)] $(n,m)=(2g+2+a,4g+1+a)$, for $a \geq 0$.
	\end{enumerate}
\end{remark}

For a one-point AG code with a given set of evaluation points, punctured codes of shorter length can be obtained by evaluating in a subset of the evaluation points. For a given complete flag of one-point AG codes, puncturing the flag will induce 
a new complete flag of one-point AG codes of reduced length. In Section 3, 
we give a characterization of isometry-dual flags in terms of maximum sparse ideals. 
And we prove that, given an isometry-dual complete flag, the isometry-dual property holds for a punctured version of the flag 
only if the number of punctured coordinates is a nongap of the Weierstrass semigroup at the defining point. In general, the isometry-dual property is not inherited if a flag is punctured in a single coordinate.

For the case of genus $g=3$, all potentially admissible pairs are located in the following table.

\[
\begin{array}{c}
\begin{array}{rccccccccccccccccccccc} \toprule
&m= &0  & & 1 &2 &3 &4 &5 &6 & &7 &8 &9 &10 &11 &12  & &13 \\ \midrule
n=1 & &\cdot & & & & & & & \\ \noalign{\medskip}
2 & & &  & &\cdot	&\cdot		&\cdot		&\cdot			&\cdot &  &\cdot \\
3 & & &  & & 			&			&\cdot		&\cdot			&\cdot  & & \cdot &\cdot \\
4 & & &  & & 			&  			& 				& &\cdot &		&\cdot &\cdot &\cdot \\ \noalign{\medskip}
5 & & &  & &			&      		& 				& & & &\cdot		&\cdot &\cdot &\cdot \\
6 & & &  & & 			&      		& 				& & & & &			&\cdot &\cdot &\cdot \\
7 & & &  & & 			&      		& 				& & & & &   		& & &\cdot &\cdot  &  \\ \noalign{\medskip}
8 & & &  & & 			&      		& 				& & & & &	& & &       & & &\cdot \\ \bottomrule
\end{array} \\ \noalign{\bigskip}
\end{array}
\]

\begin{example}\label{ex:Hermitian}

We illustrate the results for the Hermitian curve $y^3+y=x^4$ of genus $g=3$ over the field $\ff_9$. The curve has $27$ finite rational points and one point $Q$ at infinity.
The functions $x$ and $y$ have poles at $Q$ of order $3$ and $4$, respectively. The matrix $G$ for the complete flag of one-point AG codes on the curve
is of size $27 \times 27$. The rows of $G$ are labeled by monomials in $x$ and $y$ of increasing pole order. The $27$ monomials are $x^i y^j$, for $0 \leq i \leq 8$, $0 \leq j \leq 2.$
Their pole orders are $0,3,\ldots,24$ (for $j=0$), $4,7,\ldots,28$ (for $j=1$) and $8,11,\ldots,32$ (for $j=2$).
The $2g$ numbers in the range $0,1,\ldots, 32$ not occuring as geometric nongaps are $1,2,5$ and $27,30,31.$
The matrix $G$ satisfies (\ref{eq:gmg}) with $M=I$ and the complete flag of one-point AG codes of length 27 is isometry-dual.  
For the Hermitian curve of genus $g=3$, we mark actual admissible pairs by $\ast$ and the remaining pairs within the range of Theorem \ref{t:main} by a dot.

\[
\begin{array}{c}
\begin{array}{rccccccccccccccccccccc} \toprule
&m= &0 & &1 &2 &3 &4 &5 &6 & &7 &8 &9 &10 &11 &12  & &13 \\
& &1 & &- &- &2 &3 &- &4  & &5 &6 &7 &8 &9 &{10}  & &{11} \\  \midrule
n=1 & &\ast & & & & & & & \\ \noalign{\medskip}
2 & &   &   &  &\cdot &\ast   &\ast &\cdot &\cdot & &\cdot \\
3 & &   & & & & &\ast &\cdot &\ast & & \cdot &\ast \\
4 & &   & & & &  & & &\cdot & &\ast &\cdot &\ast \\  \noalign{\medskip}
5 & &   & & & &      & & & & &\cdot  &\ast &\cdot &\cdot \\
6 & &   & & & &      & & & & & &   &\cdot &\ast &\ast \\
7 & &   & & & &      & & & & &   &  & & &\ast &\ast  &  \\  \noalign{\medskip}
8 & &   & & & &      & & & & &   & & & &       &  & &\ast \\ \bottomrule
\end{array} \\ \noalign{\bigskip}
\end{array}
\]

Note that by (c) and (d) of Theorem \ref{t:main}, for $m\geq 12$, admissible pairs occur only if $n=m+1-2g$. The converse is also true, that is, for the Hermitian curve of genus 3, all pairs $(n,m)$ with $m\geq 12$ and $n=m+1-2g \leq 27$ are admissible.



\end{example}

\section{Isometry-dual flags of linear codes}

For a complete flag of codes $(C_i)_{i=0, \ldots, n}$, the isometry-dual property is defined only in terms of their inclusion of its dual flag modulo equivalence. So, we can extend it to general linear code, not necessarily induced from a curve. 

\begin{definition}
Let $A$ be a $n\times n$ matrix over $\mathbb{F}_q$ of rank $n$. The matrix $A$ is \textit{isometry-dual} if there exists a vector $\mathbf{v} \in (\mathbb{F}_q^\times)^n$ such that the matrix $A\cdot \text{diag }(\mathbf{v}) \cdot A^T$ is a anti-diagonal lower triangular matrix with nonzero anti-diagonal components, i.e.
\begin{equation}\label{mat:diag}
A\cdot \text{diag}(\mathbf{v}) \cdot A^t =
 \left(\begin{array}{ccccc}
 & & & & \star \\
 &\text{\huge 0} & & \star & \\
 & & \reflectbox{$\ddots$} & & \\
 & \star & & \text{\huge $*$} \\
\star& & & &
\end{array}\right)
\end{equation}
where all the anti-diagonal components, i.e. all $(i,j)$-th components with $i+j=n+1$, are nonzero. We call $\mathbf{v}$ a \textit{dualizing vector} and the nonzero anti-diagonal components \textit{pivots}.
\end{definition}

\begin{remark}
	Note that if a complete flag of one-point AG codes is isometry-dual then it has an isometry-dual generator matrix $G$ according to (\ref{eq:gmg}) and the diagonal matrix of a dualizing vector gives the $M$ of (\ref{eq:gmg}).
\end{remark}

\subsection{Proof of Theorem~\ref{t:main}} \label{s:MT}
Let $(C_L(D, m_iQ))_{i=0, \ldots, n}$ be a complete flag of one-point AG codes of length $n$ defined with geometric nongaps $-1 = m_0 < 0 = m_1 < \cdots < m_{n-1} < m_n = m$. The following Lemma \ref{l:mn} and Proposition \ref{p:main} and Proposition \ref{c:main} complete the proof of Theorem \ref{t:main}. 

\begin{lemma} \label{l:mn}
For an effective divisor $D$ of degree $n$, let $m$ be minimal such that $\ell(mQ) - \ell(mQ-D) = \deg D$. Then, in general,
	$m \leq n+2g-1$ and the following hold.
	\begin{enumerate}
		\item[(a)] $n=1$ if and only if $m=0$.
		\item[(b)] $n \leq m/2+1$ for $0 \leq m \leq 2g.$
		\item[(c)] $n \leq m+1-g$ for $m \geq 2g$.
	\end{enumerate}
\end{lemma}

\begin{proof}
	The upper bound is clear by the Riemann-Roch Theorem, since $\ell(mQ) - \ell(mQ-D) = \deg D$ for $m = n+2g-1$. For the second part, 
	use $\deg D \leq \ell(mQ)$. For $0 \leq m \leq 2g$, $\ell(mQ) \leq m/2+1$ by Clifford's Theorem. 
	For $m \geq 2g$, $\ell(mQ) = m+1-g$.    
\end{proof}

According to the following proposition from \cite{GMRT}, isometry-dual complete flags of one-point AG codes with $n>2g+2$ can be characterized in terms of $m$ being equal to $n+2g-1$.

\begin{proposition}[Proposition 4.3 of \cite{GMRT}]\label{p:GMRT} Suppose $n>2g+2$. Then the following are equivalent for the complete flag $(C_L(D, m_iQ))_{i=0, \ldots, n}$.
	\begin{enumerate}
		\item[(a)] The flag is isometry-dual.
		\item[(b)] $(n+2g-2)Q-D$ is a canonical divisor.
		\item[(c)] $n+2g-1 \in W^*$
	\end{enumerate}
\end{proposition}

The proof in \cite{GMRT} for the equivalence of (a) and (b) makes use of a result in \cite{MP} that requires $n>2g+2$. 
The proof for the equivalence of (b) and (c) holds in general. It follows that the equivalence of (a) and (c) holds for $n>2g+2$.
In this section, we use a different approach and we prove the equivalence of (a) and (c) directly under the weaker condition $n\geq 2g+2$.
We will also show that the weaker condition is best possible, that is, the equivalence of (a) and (c) does not hold for $n = 2g+1$. 
In fact the reduced flag of length 3 of Hermitian curve over $\mathbb{F}_{2^2}$ that was presented in Section \ref{s:Introduction} is isometry-dual with $g=1$, $n=2g+1$ and $n+2g-1 \not \in W^\ast$. \\

Properties (c) and (d) in Theorem \ref{t:main} follow from the next proposition. 



\begin{proposition}\label{p:main}
	Let $\big(C_L(D, m_iQ)\big)_{i=0, \ldots, n}$ be a complete flag of isometry-dual one-point AG codes. Let $m=m_n$. Then the following holds.
	\begin{enumerate}
		\item[(a)] $n \leq  \dfrac{m}2 + \dfrac32 < 2g+2$ for $m \leq 4g$. 
		\item[(b)] $n \geq m+1-2g$, with equality $n = m+1-2g$ for $m\geq 4g$.
	\end{enumerate}
\end{proposition}

\begin{remark}
	In the above proposition $n \geq 2g+2$ occurs only as part of case (b) in which case $m_n=n+2g-1$ and thus $n+2g-1 \in W^\ast
	= \{ m_1, m_2, \ldots, m_n \}.$ This proves that (a) implies (c) in Proposition \ref{p:GMRT} under the weaker condition $n \geq 2g+2.$  
In fact, as the next proposition shows,  the converse (c) implies (a) also holds under the condition $n \geq 2g+2.$ Thus the equivalence of (a) and (c) in Proposition~\ref{p:GMRT} holds under the weaker condition $n\geq 2g+2$. 
\end{remark}

We resume the proof of Proposition \ref{p:main} after the next proposition which proves the converse property in Theorem \ref{t:main}.

\begin{proposition}\label{c:main} For $n\geq 2g+2$, a complete flag $(C_L(D, m_iQ))_{i=0, \ldots, n}$ is isometry-dual if and only if $n+2g-1 \in W^*$.
\end{proposition}

\begin{proof}
	Note that $m=m_n$. The only if direction follows from Proposition \ref{p:main} as pointed out in the remark. 
	For the other direction, let $n+2g-1 \in W^*$. Then $m=n+2g-1$ by the Riemann-Roch Theorem. From $\dim C_L(D, (m-1)Q) =n-1$
	it follows that $\dim L((m-1)Q-D) = g$ and therefore that $(m-1)Q-D$ is a canonical divisor. Another application of the Riemann-Roch Theorem then
	shows that $C_L(D,aQ) \neq C_L(D,(a-1)Q)$ if and only if $C_L(D,(m-a)Q) \neq C_L(D,(m-a-1)Q)$, i.e., that $a \in W^\ast$ if and only if $m-a \in W^\ast$. In particular, $m_{n-i} = m-1-m_i$ and $\dim C_L(D, m_i Q) + \dim C_L(D, m_{n-i}Q) = n,$ for $0 \leq i \leq n.$ Since $L((m-1)Q-(D-P)) = L((m-1)Q-D)$, for any point $P \in D$, the code $C_L(D, (m-1)Q)$ has no words of weight one and its dual code of dimension one is therefore the span of an everywhere nonzero vector $v$. Using the vector $v$ as dualizing vector we obtain that
	\[ C_L(D, m_i Q) \perp_{\mathbf{v}} C_L(D, m_{n-i}Q)\]
	for all $0 \leq i \leq n$. Thus the flag is isometry-dual.
\end{proof}

Remark \ref{r:maxsp} in Section \ref{s:maxsp} gives an interpretation for $n+2g-1 \in W^*$ in terns of maximum sparse ideals.

For the proof of Proposition \ref{p:main} we make use of a series of lemmas. 

\begin{lemma}\label{l:pivot}
	For isometry-dual complete flag $(C_L(D, m_iQ))_{i=0, \ldots, n}$, let $u, v \in W$ such that $u+v = m$. Then  the coordinate $(u,v)$ corresponds to a pivot position.
\end{lemma}

\begin{proof}
	Let $f_u$ and $f_v$ be functions of weighted degree $u$ and $v$. Then $f_u f_v$ is of the form $\sum_{l\leq m} a_l f_l$ and $\ell(mQ) - \ell((m-1)Q) = 1$. Since the function $f_u f_v$ has degree $m$, the coefficient $a_m$ is nonzero, whence $\text{ev}_D(f_u f_v) \neq 0$.
\end{proof}

Clifford's Theorem on special divisors given as multiples of one point can be stated in terms of numerical semigroups and Dyck paths \cite{BdM}.

\begin{lemma}[Clifford's Theorem / Dyck path formulation]\label{l:Dyck}
The following inequalities hold for $1\leq a \leq 2g$.
	\begin{align*}
	&\# \{ w \in W : 1 \leq w \leq a\} \leq a/2\\
	&\# \{ w \in W : a + 1 \leq w \leq 2g \} \geq \dfrac{2g-a}{2}
	\end{align*}
Or equivalently, 
\begin{align*}
&\text{\# of nongaps in } [1,a] \leq \dfrac{a}2 \leq \text{\# of gaps in } [1,a]\\
&\text{\# of gaps in } [a+1, 2g] \leq \dfrac{2g-a}2 \leq \text{\# of nongaps in } [a+1, 2g].
\end{align*}

\end{lemma}

\begin{proof}
	The first inequality follows from the Clifford's Theorem applied to the divisor $aQ$. The second one is its complement to the fact that there are exactly $g$ Weierstrass gaps nongaps in $[1, 2g]$.
\end{proof}

\begin{theorem}[Riemann-Roch and Clifford's Theorem]\label{t:RR&C} Consdier the Riemann-Roch space $L(mQ)$ and let $n:=\ell(mQ)$. Then by Reimann-Roch theorem and Clifford's theorem, we get the following inequalities:
	\begin{align*}
	&n \geq m + 1 -g\\
	&n \leq 1 + \dfrac{m}2
	\end{align*}
\end{theorem}

According to the theorem, the possible values of $(n,m)$ are restricted to the white region in Figure~\ref{f:pair1}.

\begin{figure}\begin{center}
	\begin{tikzpicture}[scale=0.9]
\draw[->] (0,-1) -- (12,-1) node[anchor=south] {$m$};
\draw[->] (0,-1) -- (0,-8) node[anchor=east] {$n$};

\draw[smooth, samples=2, domain=5:10] plot (\x,{-\x+4}); 
\draw[smooth, samples=2, domain=-1:-6] plot (-2-2*\x, \x);  

\draw[very thick, smooth, samples=2, domain=10:12] plot (\x,{-\x+4}) node[right]{$n=m+1-g$};
\draw[thick, dotted, smooth, samples=2, domain=-6:-7] plot (-2-2*\x, \x) node[right]{$n=1+\dfrac{m}{2}$};

\draw (5,-1.1)--(5,-0.9) node[above] {$g$};
\draw (10,-1.1)--(10,-0.9) node[above] {$2g$};
\draw (9,-1.1)--(9,-0.9) node[above] {$~~2g-1$};
\draw (8,-1.1)--(8,-0.9) node[above] { $2g-2~~$};
\draw (0,-1.1)--(0,-0.9) node[above] {$0$};
\draw (-0.1, -6) -- (0.1,-6) node[left] {$g+1~$};
\draw (-0.1, -5) -- (0.1,-5) node[left] {$g~$};
\draw (-0.1, -1) -- (0.1,-1) node[left] {$1~$};
\node[above] at (6,-1) {$~~ \cdots$};

\node[draw,circle,inner sep=1.5pt,fill] at (10,-6) {};
\node[draw,circle,inner sep=1.5pt,fill] at (9,-5) {};
\node[draw,circle,inner sep=1.5pt,fill] at (8,-5) {};
\node[draw,circle,inner sep=1.5pt,fill] at (8,-4) {};

\draw[dashed] (10,-1) -- (10,-6);
\draw[dashed] (9,-1) -- (9,-5);
\draw[dashed] (8,-1) -- (8,-5);
\draw[dashed] (0, -6) -- (10, -6);
\draw[dashed] (0, -5) -- (9, -5);

\fill [opacity=0.3, gray!50, domain=5:12, variable=\x]
(5,-1) -- plot ({\x}, {-\x+4}) -- (12, -1) -- cycle;

\fill [opacity=0.3, gray!50, domain=-1:-6, variable=\x]
(0,-1) -- plot ({-2-2*\x, \x}) -- (12,-8 )-- (0, -8) -- cycle;

\end{tikzpicture}\end{center}
\caption{Dimension-degree pairs $(n,m)$ admitted by Theorem~\ref{t:RR&C}}\label{f:pair1}
\end{figure}

\begin{lemma}\label{l:divisor}
	For $m$ satisfying $2g < m \leq 4g$, the interval $[m-2g-1, 2g+1]$ contains $u, v \in W$ such that $u+v=m$.
\end{lemma}

\begin{proof}
	By Lemma~\ref{l:Dyck}, the Weierstrass nongaps are as many as Weierstrass gaps in the interval $[m-2g, 2g]$. If their numbers are not equal, by the pigeonhole principle, there exists $u, v \in [m-2g, 2g]$ such that $u+v=m$. So, suppose that there are equal number of Weierstrass gaps and nongaps in the interval $[m-2g, 2g]$. Then again by Lemma~\ref{l:Dyck}, the number $m-2g-1$ is a nongap. Then $(u,v) = (m-2g-1, 2g+1)$ satisfies the condition.
\end{proof}

\begin{proof}[Proof of Proposition~\ref{p:main}]
	In the interval $[0, m]$, there are at most $2g$ geometric nongaps. So, $n\geq m+1-2g$ holds in general. If $m\geq 4g$ then $m-2g$ is a nongap, so $(2g, m-2g)$ and $(m-2g, 2g)$ both correspond to pivot positions. Then all pivot positions are of the form $(x,y)$ with $x+y=m$. There are a total of $n=m+1-2g$ such pairs. For $2g < m \leq 4g$, consider the interval $[m-2g-1,2g+1]$. By Lemma~\ref{l:divisor}, there is a pivot $(u,v) = (m-2g+b, 2g-b)$ for some $b\geq 0$ or $(u,v) = (m-2g-1,2g+1)$ corresponds to a pivot. 
	
	If $(u,v) = (m-2g+b, 2g-b)$ corresponds to a pivot, consider the partition $\{(x,y): x < m-2g+b\}$ and $\{(x,y): y \leq 2g-b\}$ of pivots. Using the symmetry of pivot positions, the dimension $n$ is less than the number of Weierstrass nongaps. By Lemma~\ref{l:Dyck} and the partition, we get 
	\begin{align}
	n\leq \dfrac{m-2g+b-1}2 + 1 + \dfrac{2g-b}{2}+1 = \dfrac{m}2 + \dfrac{3}2.\label{eq:bound}
	\end{align}
	If $(u,v) = (m-2g-1, 2g+1)$ corresponds to a pivot, consider the partition $\{(x,y):x < m-2g\}$ and $\{(x,y) : y \leq 2g\}$, which induces the same inequality as (\ref{eq:bound}) by setting $b=0$. 
\end{proof}

\begin{remark}
According to Theorem~\ref{t:main}, possible pairs for $(n,m)$ lie in the white region of the following diagram including the solid boundaries. If $m\geq 4g+1$ then they always lie on the bolded line of $n=m+1-2g$.

\begin{figure}\begin{center}
	\begin{tikzpicture}[scale=0.8]
	\draw[->] (0,-1) -- (14,-1) node[anchor=south] {$m$};
	\draw[->] (0,-1) -- (0,-9) node[anchor=east] {$n$};
	
	\draw[smooth, samples=2, domain=6:13] plot (\x,{-\x+5}); 
	\draw[smooth, samples=2, domain=-1.5:-8] plot (-3-2*\x, \x);  
	
	\draw[very thick, smooth, samples=2, domain=13:14] plot (\x,{-\x+5}) node[below right]{$n=m+1-2g$};
	\draw[thick, dotted, smooth, samples=2, domain=-8:-8.5] plot (-3-2*\x, \x) node[right]{$n=\dfrac32+\dfrac{m}{2}$};

	\draw (6,-1.1)--(6,-0.9) node[above] {$2g$};
	\draw (11,-1.1)--(11,-0.9) node[above] {$4g-1$};
	\draw (12,-1.1)--(12,-0.9) node[above] {$4g$};
	\draw (13,-1.1)--(13,-0.9) node[above] {$4g+1$};
	\draw (0,-1.1)--(0,-0.9) node[above] {$0$};
	\draw (-0.1, -6) -- (0.1,-6) node[left] {$2g~$};
	\draw (-0.1, -7) -- (0.1,-7) node[left] {$2g+1~$};
	\draw (-0.1, -8) -- (0.1,-8) node[left] {$2g+2~$};
	\draw (-0.1, -1) -- (0.1,-1) node[left] {$1~$};
	\draw (-0.1, -2) -- (0.1,-2) node[left] {$2~$};
	\node[above] at (8,-1) {$~~ \cdots$};
	
	\node[draw,circle,inner sep=1.5pt,fill] at (11,-7) {};
	\node[draw,circle,inner sep=1.5pt,fill] at (12,-7) {};
	\node[draw,circle,inner sep=1.5pt,fill] at (13,-8) {};
	\node[draw,circle,inner sep=1.5pt,fill] at (11,-6) {};

	\draw[dashed] (11,-1) -- (11,-7);
	\draw[dashed] (12,-1) -- (12,-7);
	\draw[dashed] (13,-1) -- (13,-8);
	\draw[dashed] (0, -7) -- (12, -7);
	\draw[dashed] (0, -8) -- (13, -8);
	\draw[dashed] (0,-6) -- (11,-6);
	
	\fill [opacity=0.3, gray!50, domain=5:14, variable=\x]
	(6,-1) -- plot ({\x}, {-\x+5}) -- (14, -1) -- cycle;
	
	\fill [opacity=0.3, gray!50, domain=-1:-8, variable=\x]
	(0,-1.5) -- plot ({-3-2*\x, \x}) -- (14,-9) -- (0, -9) -- cycle;
	
	\end{tikzpicture}\end{center}
\caption{Dimension-degree pairs $(n,m)$ admitted by Proposition~\ref{p:main}}\label{f:pair2}
\end{figure}
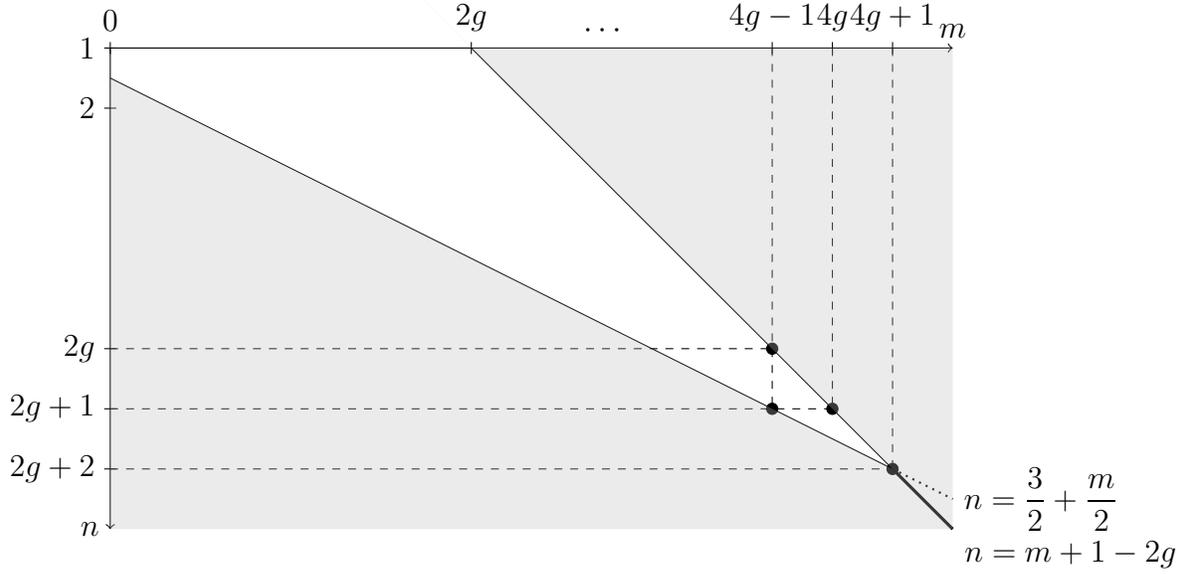
	
\end{remark}

\begin{example}
Let $\mathcal{X}$ be a Hermitian curve defined by the affine equation $y^2+y=x^3$ over $\mathbb{F}_4 = \mathbb{F}_2(\alpha)$ with $\alpha^2 + \alpha+1=0$. Then the genus of $\mathcal{X}$ is $1$. 
\begin{enumerate}
\item $m=4g-1= 3$ and $n=2g+1=3$.

Let $D'= \{ (0,1), (\alpha, \alpha), (\alpha^2, \alpha^2)\} \sim 2 P_{\infty} + (0,0)$. Functions corresponding to the geometric nongaps of the compelte flag $(C_L(D', m_iQ))_{i=0,1,2,3}$ are $1$, $x$, and $y$. Then we get the following generator matrix $G$:
\[
\begin{array}{cccc}
	\toprule
	 &(0,1) & (\alpha, \alpha) & (\alpha^2, \alpha^2)\\
	\midrule
	1~~ & 1 & 1 & 1\\
	x~~ & 0 & \alpha & \alpha^2\\
	y~~ & 1 & \alpha & \alpha^2\\
	\bottomrule
\end{array}
\]
with a dualizing vector $\mathbf{v}=(1,\alpha^2,\alpha)$. Then
\[
G \cdot \text{diag}(\mathbf{v})\cdot G^T = 
\left[ \begin{array}{ccc}
0 & 0 & 1\\
0 & 1 & 1\\
1 & 1 & 0
\end{array}
\right]
\]

\item $m=4g-1=3$ and $n=2g=2$.

Let $D' = \{ (0,0) , (0,1)\}$. For a complete flag $(C_L(D', m_iQ))_{i=0,1,2}$, functions with only pole at $P_\infty$ of order $m_i$ are $1$ and $y$. Then we get the following generator matrix:
\[
\begin{array}{ccc}
\toprule
&(0,0) & (0,1) \\
\midrule
1~~ & 1  & 1\\
y~~ & 0 & 1\\
\bottomrule
\end{array}
\]
with a dualizing vector $\mathbf{v}=(1, 1)$. Then
\[
G \cdot \text{diag}(\mathbf{v})\cdot G^T = 
\left[ \begin{array}{cc}
0 & 1\\
1 &  0
\end{array}
\right]
\]

\end{enumerate}

\end{example}

\begin{example}
The Hemitian curve $\mathcal{X}$ over the field $\mathbb{F}_9$ is defined by the equation $y^3+y = x^4$. This curve has the genus $3$. Let $\mathbb{F}_9 = \mathbb{F}_3[\alpha]$ with $\alpha^2 - \alpha- 1 = 0$.

\begin{enumerate}

\item $m=4g-1=11$ and $n=2g+1=7$.

Let $ D' = \{  (0, \alpha^2), (0, \alpha^6), (1, 2), (\alpha, 1), (\alpha^3, 1), (\alpha^5, \alpha^7), (\alpha^7, \alpha^5) \} \sim 6 P_\infty + (0,0)$. For the complete flag $(C_L(D', m_iP_\infty))_{i=0, \ldots, 7}$, corresponding functions with pole order $m_i$ at $P_\infty$ for $i=1, \ldots, 7$ are $1,~ x,~ y,~ x^2,~ xy,~ y^2$ and $xy^2$. Then we will get the following generator matrix $G$.

\[
\begin{array}{cccccccc}
\toprule
&(0, \alpha^2) &(0, \alpha^6) & (1, 2) & (\alpha, 1) &(\alpha^3, 1) & (\alpha^5, \alpha^7) &(\alpha^7, \alpha^5)\\
\midrule
1~~ & 1 & 1 & 1 & 1 & 1 & 1 & 1\\
x~~ & 0 & 0 & 1 & \alpha & \alpha^3 & \alpha^5 & \alpha^7\\
y~~ & \alpha^2 & \alpha^6 & 2 & 1 & 1 & \alpha^7 & \alpha^5\\
x^2~~ & 0 & 0 & 1 & \alpha^2 & \alpha^6 & \alpha^2 & \alpha^6\\
xy~~ & 0 & 0 & 2 & \alpha & \alpha^3 & \alpha^4 & \alpha^4 \\
y^2~~ & \alpha^4 & \alpha^4 & 1 & 1 & 1 & \alpha^6 & \alpha^2 \\
xy^2~~& 0 & 0 & 1 & \alpha & \alpha^3 & \alpha^3 & \alpha\\
\bottomrule
\end{array}
\]

With a dualizing vector $\mathbf{v}=(1,1,2,\alpha^7, \alpha^5, \alpha, \alpha^3)$, we get
\[ G \cdot \text{diag}(\mathbf{v})\cdot G^T = 
\left[  
\begin{array}{ccccccc}
0 & 0 & 0 & 0 & 0 & 0 & 1 \\
0 & 0 & 0 & 0 & 0 & 1 & 0 \\
0 & 0 & 0 & 0 & 1 & 1 & 2 \\
0 & 0 & 0 & 1 & 1 & 0 & 1 \\
0 & 0 & 1 & 1 & 0 & 2 & 1 \\
0 & 1 & 1 & 0 & 2 & 1 & 0 \\
1 & 0 & 2 & 1 & 2 & 0 & 1
\end{array}
\right]
\]

\item $m=4g=12$ and $n=2g+1=7$.

Let 
$D' = \{  (1,\alpha), (1,\alpha^3), (1,2), (\alpha,1), (\alpha^3,1), (\alpha^5,1), (\alpha^7,1) \} \sim 7 P_\infty$. For the complete flag $(C_L(D', m_iP_\infty))_{i=0, \ldots, 7}$, corresponding functions with pole order $m_i$ at $P_\infty$ for $i=1, \ldots, 7$ are $1,~ x,~ y,~ x^2,~ y^2,~ x^3$ and $x^4$. Then the generator matrix $G$ is
\[
\begin{array}{cccccccc}
\toprule
& (1,\alpha)& (1,\alpha^3)& (1,2)& (\alpha,1)& (\alpha^3,1)& (\alpha^5,1)& (\alpha^7,1)\\
\midrule
1~~ & 1 & 1 & 1 & 1 & 1 & 1 & 1\\
x~~ & 1 & 1 & 1 & \alpha & \alpha^3 & \alpha^5 & \alpha^7\\
y~~ & \alpha & \alpha^3 & 2 & 1 & 1 & 1 & 1 \\
x^2~~ & 1 & 1 & 1 & \alpha^2 & \alpha^6 & \alpha^2 & \alpha^6\\
y^2~~ & \alpha^2 & \alpha^6 & 1 & 1 & 1 & 1 & 1\\
x^3~~ & 1 & 1 & 1 & \alpha^3 & \alpha & \alpha^7 & \alpha^5\\
x^4~~ & 1 & 1 & 1 & \alpha^4 & \alpha^4 & \alpha^4 & \alpha^4 \\
\bottomrule
\end{array}
\]
With a dualizing vector $\mathbf{v} = (\alpha^5,\alpha^7,2,\alpha^2,\alpha^6,\alpha^7,\alpha^5)$, we get \[G \cdot \text{diag}(\mathbf{v}) \cdot G^T =
\left[
\begin{array}{ccccccc}
0 & 0 & 0 & 0 & 0 & 0 & 1\\
0 & 0 & 0 & 0 & 0 & 1 & 1\\
0 & 0 & 0 & 0 & 1 & 0 & 1\\
0 & 0 & 0 & 1 & 0 & 1 & 1\\
0 & 0 & 1 & 0 & 1 & 0 & 1\\
0 & 1 & 0 & 1 & 0 & 1 & 1\\
1 & 1 & 1 & 1 & 1 & 1 & 0
\end{array}\right]
\]

\end{enumerate}

\end{example}

\subsection{Some worked examples}

Note that in the first section, we classified all admissible pairs for
the Hermitian curve of genus 3 for properly chosen $D$ and $Q$. In
this section we consider two examples of complete flag of one-point AG
codes defined over curves of genus 3 and find admissible
pairs. Another example, the Klein curve, will be delt in the later
section. 
  There exist exactly $4$ numerical semigroups of genus $3$.
  Next, for each of these semigroups we associate an example
  of a curve having this semigroup as the Weierstrass semigroup at some point of the curve.

\[
\begin{array}{lll}
	\toprule
	\text{Nongap structure}	&~~~	\text{Type of curve}	&~~~ \text{Reference}\\
	\midrule
	\{0, 3, 4, 6, 7, 8, 9, \ldots\} &~~~ \text{Hermitian curve} &~~~ \text{Example \ref{ex:Hermitian}}\\
	\{0,2,4,6,7,8,9,\ldots\} &~~~ \text{Hyperelliptic curve} &~~~ \text{Example \ref{ex:Hyperelliptic}} \\
	\{0, 4, 5, 6, 7, 8, 9, \ldots\} &~~~  \text{Hyperelliptic curve} &~~~\text{Example~\ref{ex:new curve}}\\
	\{0, 3, 5, 6, 7, 8, 9, \ldots\} &~~~ \text{Klein curve} &~~~ \text{Example \ref{ex:Klein}}  \\
	\bottomrule
\end{array}
\]

Theorem~\ref{t:main} restricts the occurance of admissible pairs in certain range. With the above four curves, we will see that all possible admissible pairs are actually obtained from curves.

\begin{example}[Hyperelliptic curve]\label{ex:Hyperelliptic}
	Consider a hyperelliptic curve given by the equation $y^2 = x^7+x^6-x$ over $\mathbb{F}_7$. The curve has 13 rational points whose coordinates are given by $(0,0)$, $(a,1)$ and $(a,-1)$ for each $a\in \mathbb{F}_7$. It also has one point at infinity, say $Q$. The functions $x$ and $y$ have poles at $Q$ of order $2$ and $7$, respectively. The semigroup of Weierstrass nongaps at $Q$ is given by $\{0, 2, 4, 6, 7, 8, 9, \ldots\}$. The following table shows all admissible pairs $(n,m)$ for isometry-dual flags $C_L(D, m_iQ)_{i=1, \ldots, n}$.	
	
	\[
\begin{array}{c}
\begin{array}{rccccccccccccccccccccc} \toprule
&m= &0 & &1 &2 &3 &4 &5 &6 & &7 &8 &9 &10 &11 &12  & &13 \\
& &1 & &- &2 &- &3 &- &4 & &5 &6 &7 &8 &9 &{10}  & & {11} \\  \midrule
n=1 & &\ast & & & & & & & \\ \noalign{\medskip}
2 & & &  & &\ast & \cdot  & \cdot &\cdot &\cdot & &\ast \\
3 & & &  & & &  &\ast &\cdot &\cdot & &\cdot &\cdot \\
4 & & &  & & &  & & &\ast & &\cdot &\cdot &\ast \\  \noalign{\medskip}
5 & & &  & & &  & & & & &\ast &\ast &\cdot &\cdot \\
6 & & &  & & &  & & & & & &   &\ast &\ast &\ast \\
7 & & &  & & &  & & & & & &   & & &\cdot &\cdot & &  \\  \noalign{\medskip}
8 & & &  & & &  & & & & & &   & & &       & & &\cdot \\ \bottomrule
\end{array} \\ \noalign{\bigskip}
\end{array}
\]
\end{example}

\begin{example}[Hyperelliptic Curve 2]\label{ex:new curve}
	Consider a curve in the projective 4 space over the field $\mathbb{F}_2$ defined by the following affine equations:
	\begin{align*}
&	uv + uw + v^2=0\\
&	u^3 + w^2 + w=0\\
&	u^3 + u^2v + wx=0\\
&	u^3 + u^2 + uv^2 + x^2 + x=0
	\end{align*}
It is an irreducible curve of genus 3 with three points on the curve
$P_1=(0,0,0,0)$, $P_2=(0,0,0,1)$ and $P_3=(0,0,1,0)$. For a point at
infinity $Q(u,v,w,x,t)=(0,0,0,1,0)$, each function $u$, $v$, $w$, and
$x$ has a pole only at $Q$ with the following orders:
\begin{align*}
\text{ord}_Q(u) = -4\\
\text{ord}_Q(v)=-5\\
\text{ord}_Q(w)=-6\\
\text{ord}_Q(x)=-7
\end{align*}

We will get the following evaluation matrix of the functions $1$, $u$, $v$, $w$, and $x$.
\[
\begin{array}{cccc}
\toprule
& P_1 & P_2 & P_3 \\
\midrule
1~~ & 1 & 1 & 1 \\
u~~ & 0 & 0 & 0 \\
v~~ & 0 & 0 & 0 \\
w~~ & 0 & 0 & 1\\
x~~ & 0 & 1 & 0\\
\bottomrule
\end{array}\]
Then $(n,m) = (2,6)$ is an admissible pair, since
\[
\begin{array}{ccc}
\toprule
& P_1 & P_3\\
\midrule
1 & 1 & 1\\
w & 0 & 1\\
\bottomrule
\end{array}
\]
is a self-dual matrix.

Note that with a different model by setting $x = \dfrac{v}{w}$ and $y = \dfrac{v^4}{w^3 u^2}$, the curve is given by 
\[y^2+x^2 y+y = x^7 + x^6.\]
So, it is hyperelliptic. 

\end{example}

\section{Maximum sparse ideals and the inheritance of isometry-dual condition} \label{s:maxsp}

We can see from (d) of Theorem~\ref{t:main} that an admissible pair $(n,m)$ can occur for $n>2g+1$ with $m=n+2g-1$ if we have enough rational points. A question arises here is when a flag of isometry-dual one point AG codes can be realized as a punctured sub-flag of another isometry-dual flag. In this section, we will show that not arbitrary puncturing will induce an isometry-dual flag but it is necesary that the number of punctures points should be in the Weierstrass semigroup of the point $Q$, at which the one-point AG codes are defined. To this end, we introduce the concept of a maximum sparse ideal in the theory of numerical semigroup and then use this to the property of isometry-dual flag and the Weierstrass semigroup of the defining point $Q$.

\subsection{Maximum sparse ideals} \label{s:maxsparse}

Let ${\mathbb N}_0={\mathbb N}\cup \{0\}$. A {\em numerical semigroup} $S$ is a subset of ${\mathbb N}_0$ that contains $0$, is closed under addition and has a finite complement in ${\mathbb N}_0$. For instance, the next set is a numerical semigroup.
$$S=\{0,3,5,6,7,8,9,10, 11, 12, \dots\}$$

The {\em genus} $g$ of a numerical semigroup $S$ is the number $g=\#{\mathbb N}_0\setminus S$.
The {\em conductor} $c$ of $S$ is the smallest integer such that $c+{\mathbb N}_0\subseteq S$.
In the previous example, the genus is $3$ and the conductor is $5$. We call the elements of $S$ by {\em nongaps} and those of $\mathbb{N}_0 \backslash S$ by {\em gaps}.

Denote the elements of $S$, in increasing order, by $\lambda_0=0,
\lambda_1,\lambda_2,\dots$ and, for each $i\geq 0$, define $D(i)=S\cap(\lambda_i-S)$ or, equivalently, define $D(i)$ as the set of nongaps $\lambda_j$ such that $\lambda_i-\lambda_j\in S$. In the running example, $D(7)=\{0,3, 5, 7, 10\}$. 

An {\em ideal} $I$ of a numerical semigroup $S$ is a subset of $S$
such that $I+S\subseteq I$. We say that $I$ is a {\em proper} ideal of $S$ if $I\neq S$. For instance, the next set is an ideal of the previous numerical semigroup.
$$I=\{6, 8, 9, 11, 12, 13, 14, \ldots \}$$

The largest integer not belonging to
an ideal is called the {\em Frobenius number} of the ideal.
The Frobenius number of the previous ideal is $10$.
The next bound on the Frobenius number of an ideal is proved in \cite{BLV}.
\begin{lemma}\label{l:fnbound}The {\em Frobenius number} of an ideal $I$ of a numerical semigroup $S$ of genus $g$
  is at most $2g-1+\# (S\setminus I)$.
\end{lemma}
The ideals whose Frobenius number attains this bound will be called {\it maximum sparse ideals}.
In the previous example, the bound $2g-1+\#(S\setminus I)$ is $5+\#\{0, 3, 5, 7, 10\}=10$ which coincides with the Frobenius number. Hence, the ideal $I$ is maximum sparse.

The next lemma characterizes maximum sparse ideals. It is also proved in~\cite{BLV}.
\begin{lemma}\label{l:maxsparse} Given a numerical semigroup $S$ with enumeration $\lambda_0,\lambda_1,\lambda_2,\dots$, let $G(i)$ be the number of pairs of gaps adding up to $\lambda_i$. A proper ideal $I$ of $S$ is maximum sparse if and only if $I=S\setminus D(i)$ for some $i$ with $G(i)=0$.
\end{lemma}
For a maximum sparse ideal $I$ we call its Frobenius number the {\em leader} of the ideal. By the previous lemma, the leader is a nongap $\lambda_i$ such that it is not the sum of any two gaps. Furthermore, the ideal $I$ is then $I=S\setminus D(i)$.

Let us check that the result of Lemma~\ref{l:maxsparse} is satisfied in the running example. On one hand, we already saw that the ideal is maximum sparse. We can check now the equivalent condition. Indeed, $I=S\setminus\{0,3, 5, 7, 10\}$, where $\{0,3,5, 7, 10\}=D(7)$, while $G(7)=0$ because there is no pair of gaps adding up to $\lambda_{7}=10$. In this case, $10$ is the leader of $I$. Note that in the running example, $S\setminus D(10),~ S\setminus D(13),~ S\setminus D(16), \dots$ are also maximum sparse and it will be clear in Example~\ref{ex:Klein}.

\begin{lemma}
The leaders of proper maximum sparse ideals are always at least as large as the conductor.
\end{lemma}

\begin{proof}
  By Lemma~\ref{l:fnbound}, the Frobenius number of $S$ is $2g-1+\#(S\setminus I)$. Since the ideal is proper, $\#(S\setminus I)\geq 1$, and so $2g-1+\#(S\setminus I)\geq 2g\geq c$.
\end{proof}

Next theorem states the relationship between leaders of maximum sparse ideals of a given numerical semigroup when the ideals satisfy inclusion relationships.

\begin{theorem}\label{t:incl}
  For two proper maximum sparse ideals $I,I'$ of a numerical semigroup $S$ with leaders $\lambda_i,\lambda_{i'}$, the following are equivalent:
  \begin{enumerate}
  \item $I'\supseteq I$
  \item $S\setminus I'\subseteq S\setminus I$
  \item $D(i')\subseteq D(i)$
  \item $\lambda_i-\lambda_{i'}\in S$
  \item $\#(S\setminus I)-\#(S\setminus I')\in S$
  \end{enumerate}
\end{theorem}

\begin{proof}
  It is obvious that statements (1) and (2) are equivalent.

  By Lemma~\ref{l:maxsparse}, since $I,I'$ are proper maximum sparse ideals, $D(i)=S\setminus I$ and $D(i')=S\setminus I'$. Hence, statement (2) and statement (3) are equivalent.

  Statement (3) is equivalent to $\lambda_{i'}\in D(i)$ which, in turn, is equivalent to statement (4).
  
  Statements (4) and (5) are equivalent since, by Lemma~\ref{l:fnbound}, $\lambda_i=2g-1+\#(S\setminus I)$ and $\lambda_{i'}=2g-1+\#(S\setminus I')$. Hence, $\lambda_i-\lambda_{i'}=\#(S\setminus I)-\#(S\setminus I')$.
  \end{proof}

\begin{remark}
Suppose that $\lambda_i$ is the leader of a maximum sparse ideal and that $\lambda_{i'}$, which is at least the conductor, satisfies $\lambda_i-\lambda_{i'}\in S$. This does not imply that $\lambda_{i'}$ is the leader of any maximum sparse ideal, unless $D(i')=0$.
\end{remark}

The next lemma shows that the leaders of the sparse ideals of a numerical semigroup constitute in turn another ideal of the numerical semigroup.

\begin{lemma} 
The set $L$ of non-zero nongaps $\lambda_i$ such that $G(i)=0$ is an ideal of $S$.
\end{lemma}
  
\begin{proof}
First of all, notice that the nongaps smaller than the conductor do not satisfy $G(i)=0$. Indeed, if $\lambda_i<c$, there is a gap $a<\lambda_i$ with $\lambda_i-a<\lambda_1$ because otherwise $\lambda_i$ would be larger than the conductor. Now, $\lambda_i-a<\lambda_1$ is a positive gap which, together with $a$ adds up to $\lambda_i$. So, $G(i)\neq 0$. So, all the elements in $L$ are at least equal to the conductor.
  
We need to prove that if $\lambda_i\in L$ then $\lambda_i+\lambda_j\in L$ for any $\lambda_j\in S$. Assume that $\lambda_j\neq 0$. Let $k$ be such that $\lambda_i+\lambda_j=\lambda_k$. Suppose that $\lambda_k\not\in L$ and so that $G(k)\neq 0$. Then there are two gaps $a$, $b$ such that $a+b=\lambda_k$. Note that both $a, b < \lambda_i = \lambda_k - \lambda_j$ since $\lambda_i$ is greater or equal to $c$. From $a+b=\lambda_k$, we have $\lambda_j < a, b < \lambda_i$. Then $a-\lambda_j$ is a gap of $S$ since, otherwise, $a=(a-\lambda_j)+\lambda_j\in S+S\subseteq S$.
In  particular, $(a-\lambda_j)+b$ is a sum of two gaps which adds up to $a+b-\lambda_j=\lambda_k-\lambda_j=\lambda_i$, a contradiction to $G(i)=0$.
\end{proof}

\subsection{Puncturing sequences of isometry-dual one-point AG codes}

Let $\mathcal{X}$ be a smooth absolutely irreducible projective curve $\mathcal{X}$ of genus $g$ over a finite field. Let
\[C_i := C_L(D, m_iQ)_{i=0, \ldots, n}\]
 be a complete flag of one-point AG codes. Recall that $W$ is the Weierstrass semigroup at $Q$ and $W^*=\{m_1, \ldots, m_n\}$ is the set of geometric nongaps. Also, note that $W^*\subset W$.

Next lemma is stated in different words in \cite[Corollary 3.3.]{GMRT} with the condition $n>2g+2$. The same proof is still valid with $n\geq 2g+2$.

\begin{lemma}The set $W\setminus W^*$ is an ideal of $W$.
\end{lemma}

\begin{proof}
	Let $i\in W\backslash W^*$ and $j\in W$. Then by statements (\ref{cr1}) and (\ref{cr2}) of Lemma~\ref{lemmacr1cr2}, we have the following relations.
	\begin{align*}
	&L(iQ) \gneq L((i-1)Q) \\
	&L(iQ-D) \gneq L((i-1)Q-D)\\
	&L(jQ) \gneq L((j-1)Q-D)
	\end{align*}
	Let $g\in L(iQ-D)\backslash L((i-1)Q-D)$ and $f\in L(jQ)\backslash L((j-1)Q)$. Then $fg \in L((i+j)Q-D) \backslash L((i+j-1)Q-D)$. Thus $i+j$ does not satisfy the statement (\ref{cr2}), so it is not in $W^*$. 
\end{proof}

\begin{remark} \label{r:maxsp}
	Note that $n+2g-1 \in W^*$ is equlvalent to $W \setminus W^*$ being maximum sparse. Indeed, if $n+2g-1 \in W^*$, it is the maximum in $W^*$ by applyig Riemann-Roch to $C_L(D, (n+2g-1)Q)$. Then the Frobenius number $F$ of $W\setminus W^*$ satisfies $F=n+2g-1$ and $| W\setminus (W\setminus W^*) | = |W^*| = n$. Conversely, the ideal $W\setminus W^*$ being maximum sparse implies that $F=2g-1+|W^*| = n+2g-1$. Then $F \in W^*$. 
\end{remark}

Consider a complete flag of one-point AG codes $(C_L(D, m_iQ))_{i=0, \ldots, n}$. Using a notational abuse, we can consider $D$ as an ordered set. Let $D' \subseteq D$ with $|D'| = s \leq n$. By reordering the rational points $P_1, \ldots, P_n$, if necesary,  we may assume that $D'=P_1+P_2 + \cdots + P_s$. Construct a complete flag of one-point AG codes $(C_L(D_s, m'_iQ))_{i=0, \ldots, s}$ and call it a {\em punctured flag}. Let $W'$ denote the set of geometric nongaps of the induced complete flag. 

\begin{lemma}
	In the above construction we have $W' \subset W^*$.
\end{lemma}

\begin{proof}
  Let $\lambda:=m'_i \in W'$ for some $i$. It is enough to show
  $C_L(D, (\lambda-1)Q) < C_L(D, \lambda Q)$. Let $f\in L(\lambda
  Q)\setminus L((\lambda-1)Q)$. Then $\text{ev}_D(f)$ is not generated
  by $\text{ev}_D(g)$'s for $\text{ord}_Q(g) > -\lambda$ since linear combination of all such $\text{ev}_D(g)$ cannot generate the $\text{ev}_D(f)$ in its $D'$ positions.
\end{proof}

\begin{theorem}\label{t:inh}
	Assume that $n, s \geq 2g+2$. Let $(C_L(D, m_iQ))_{i=0, \ldots, n}$ be an isometry-dual flag of one-point AG codes. If the punctured flag $(C_L(D', m'_iQ))_{i=0, \ldots, s}$ is isometry-dual then $n-s \in W$.
\end{theorem}

\begin{proof}
Note that $W'\subseteq W^*$ by the previous Lemma. Then by Theorem~\ref{t:incl}, we get $n-s \in W$. 
\end{proof}

\begin{remark}
	Note that the converse of the above theorem does not hold in general. That is, there is no gaurantee that the induced flag will also have the isometry-dual property just by puncturing $n-s$ random points from $D$.
\end{remark}

One conclusion of Theorem~\ref{t:inh} is that, from an isometry-dual flag of one-point AG codes, to get a punctured flag with isometry-dual condittion, one needs to take out a number
of evaluation points at least equal to the multiplicity (smallest
non-zero nongap) of the Weierstrass semigroup of $Q$.
The next example shows this property in Klein curve. From an isometry-dual flag of length 23, by properly puncturing three points repeatedly, we can get induced isometry-dual flags of length 20, 17, 14, and all the way down to of length 2.

\begin{example}[Klein Curve]\label{ex:Klein}
	The Klein curve $\mathcal{X}$ is given by $X^3Y + Y^3 Z + Z^3 X=0$ in a projective plane. There are 3 $F_2$-rational poins of $\mathcal{X}$, namely, $Q_1 = (1:0:0)$, $Q_2 = (0:1:0)$, and $Q_3= (0:0:1)$. Let $\mathbb{F}_8 = \mathbb{F}_2(w)$ such that $w^3 + w + 1= 0$. Then the $\mathbb{F}_8$-rational points on $\mathcal{X}$, which are not $Q_i$ for $i=1,2,3$ can be expressed in the following way:
	
\[	
\begin{array}{lll}
	P_1=( 1, w, 1 ),
	&P_2=( 1, w^2, 1 ),
	&P_3=( 1, w^4, 1 ),\\
	P_4=( 1, 1, w ),
	&P_5=( 1, w^4, w ),
	&P_6=( 1, w^5, w ),\\
	P_7=( 1, 1, w^2 ),
	&P_8=( 1, w, w^2 ),
	&P_9=( 1, w^3, w^2 ),\\
	P_{10}=( 1, w^3, w^3 ),
	&P_{11}=( 1, w^4, w^3 ),
	&P_{12}=( 1, w^6, w^3 ),\\
	P_{13}=( 1, 1, w^4 ),
	&P_{14}=( 1, w^2, w^4 ),
	&P_{15}=( 1, w^6, w^4 ),\\
	P_{16}=( 1, w^2, w^5 ),
	&P_{17}=( 1, w^3, w^5 ),
	&P_{18}=( 1, w^5, w^5 ),\\
	P_{19}=( 1, w, w^6 ),
	&P_{20}=( 1, w^5, w^6 ),
	&P_{21}=( 1, w^6, w^6 )
\end{array}
\]
	
	The above points are chosen in a way that $Q_2$, $P_{3i+1}$, $P_{3i+2}$, and $P_{3i+2}$ are colinear for $i=0, \ldots, 6$. The curve $\mathcal{X}$ satisfies the following properties:
	
	\begin{enumerate}
		\item The genus of $\mathcal{X}$ is 3.
		\item There are 24 $\mathbb{F}_8$-rational points, all of which are flexpoints.
		\item Denote $T_P$ the tangent line of the curve $\mathcal{X}$ at $P$. Then $I(P, \mathcal{X} \cap T_P) = 3$ for all 24 $\mathbb{F}_8$-rational points $P$ of $\mathcal{X}$. Then there is one point in $\mathcal{X}\cap T_P$ other than $P$. Denote this point by $P'$. Note that $I(P', \mathcal{X}\cap T_P) = 1$. 
		\item Write $P'' = (P')'$. Then $(P'')' = P$, that is $I(P, \mathcal{X}\cap T_{P''})=1$.
		
		\item With the previous notation for $P$, $P'$, $P''$, the canonical divisor $L$ satisfies
		\[ L \sim 3P + P' \sim 3P' + P'' \sim 3P'' + P\]
		for any $\mathbb{F}_8$-rational point $P$.
		
		\item For any $F_8$-rational points $R_1$ and $R_2$ such that $R_1 \neq R_2'$ and $R_2 \neq R_1'$, consider a line through $R_1$ and $R_2$. Then then line meets $\mathcal{X}$ at 4 distinct points including $R_1$ and $R_2$. Denote the other 2 points by $R_3$ and $R_4$. 
		
		\item With the previous notation for $R_1$, $R_2$, $R_3$, and $R_4$, the canonical divisor $L$ satisfies
		\[ L\sim R_1 + R_2 + R_3 + R_4\]
		
	      \item
The projective plane over ${\mathbb F}_8$ contains $73$ lines. Of these lines, $24$ lines intersect the Klein curve three times in a rational point and once in a different rational point, $42$ lines pass through four rational points, and the remaining seven lines do not contain rational points. Counting lines in two different ways it is verified that ${24 \choose 2} = 24 \cdot 1 + 42 \cdot 6$.
		
		\item The Weierstrass semigroup at a point $Q_3$ is $\{ 0, 3, 5, 6, 7, 8, 9, \ldots\}$.
	\end{enumerate}

	Let $Q=Q_3 = (0:0:1)$. Then $Q' = Q_2=(0:1:0)$ and $Q''=Q_1=(1:0:0)$. Consider one-point AG codes at $Q$. 
	
	Define $D_2 := Q' + Q''$. Then 
	\begin{align*}
	K+ D_2  & \sim 3Q + Q' + Q' + Q'' = 3Q + (3Q' + Q'') - Q'\\
	& \sim 3Q + K - Q'\\
	& \sim 6Q.
	\end{align*}
	The geometric nongaps of the complete flag $(C_L(D_2, m_iQ))_{i=0,1,2}$  are  $\{0, 7\}$ with corresponding functions

	\begin{align*}
	1\in L(0Q)\backslash L(0Q-D_2)\\
	Y^2Z/X^3 \in L(7Q)\backslash L(7Q-D_2).
	\end{align*}
	Then we get an isometry-dual flag of length 2 with generator matrix $G_2$
	\[\begin{array}{ccc}
	\toprule
	~~& Q'	& Q'' \\
	\midrule
	1~~ & 1	& 1 \\
	Y^2Z/X^3~~ & 1	& 0\\
	\bottomrule
	\end{array}
	\]
	
	Define $D_5 := Q'+Q''+P_1 + P_2 + P_3$. Note that $P_1$, $P_2$, $P_3$, and $Q'$ are colinear. Then
	\begin{align*}
	K + D_5 & \sim 2K + Q'' \\
	& \sim 6Q + 2Q' + Q'' = 6Q + (3Q' + Q'') - Q'' \\
	& \sim 6Q + K - Q'\\
	& \sim 6Q + 3Q + Q' - Q' = 9Q
	\end{align*}
	and $D_5 - D_2 \sim 3Q$.  The geometric nongaps of the complete flag $(C_L(D_5, m_iQ))_{i=0, \ldots, 5}$ are $\{0, 3, 5, 7, 10\}$.
The functions with the corresponding pole orders are
	\begin{align*}
	Z/X \in &L(3Q)\backslash L(3Q-D_5) \\
	YZ/X^2 \in &L(5Q)\backslash L(5Q-D_5)\\
	(YZ/X^2)^2 \in &L(10Q)\backslash L(10Q-D_5).
	\end{align*}
We get an isometry-dual flag of length 5 with generator matrix $G_5$
	\[
	\begin{array}{cccccc}
	\toprule
	&Q'&Q''&P_1&P_2&P_3 \\
	\midrule
	1~~&1&1&1&1&1 \\
	Z/X~~&0&0&1&1&1 \\
	YZ/X^2~~&0&0&w&w^2&w^4 \\
	Y^2Z/X^3~~&1&0&w^2&w^4&w \\
	(YZ/X^2)^2~~&0&0&w^2&w^4&w \\
	\bottomrule
	\end{array}
	\]
	
	Then the dualizing vector of the above matrix is $(1, 1, w, w^2, w^4)$. \\
	
	Define $D_8 := D_5 + P_4 + P_5 + P_6$. Then
	\begin{align*}
	K+D_8 & = K+D_5 + P_4 + P_5 + P_6 \\
	& \sim 9Q + K - Q' \\
	& \sim 12Q
	\end{align*}
	with $D_8 - D_5 \sim 3 Q$.	The geometric nongaps of the flag $(C_L(D_8, m_iQ))_{i=0, \ldots, 8}$ are $\{0, 3, 5, 6, 7, 8, 10, 13\}$. The functions of pole orders corresponding to geometric nongaps are
	\begin{align*}
	(Z/X)^2 \in & L(6Q)\backslash L(6Q-D_8) \\
	YZ^2/X^3 \in & L(8Q)\backslash L(8Q-D_8) \\
	Y^2Z^3/X^5 \in & L(13Q) \backslash L(13Q-D_8)
	\end{align*}
	
Then we get an isometry-dual complete flag $(C_L(D_8, m_iQ))_{i=0, \ldots, 8}$ with generator matrix $G_8$
	\[
	\begin{array}{ccccccccc}
	\toprule
	&Q'&Q''&P_1&P_2&P_3&P_4&P_5&P_6 \\
	\midrule
1~~&  			1&   1&   1&   1&   1&   1&   1&   1\\
Z/X~~&  		0&   0&   1&   1&   1&   w&   w&   w\\
YZ/X^2~~&  		0&   0&   w& w^2& w^4&   w& w^5& w^6\\
(Z/X)^2~~&  	0&   0&   1&   1&   1& w^2& w^2& w^2\\
Y^2Z/X^3~~& 	1&   0& w^2& w^4&   w&   w& w^2& w^4\\
YZ^2/X^3)~~&    0&   0&   w& w^2& w^4& w^2& w^6&   1\\
(YZ/X^2)^2~~&   0&   0& w^2& w^4&   w& w^2& w^3& w^5\\
Y^2Z^3/X^5~~&   0&   0& w^2& w^4&   w& w^3& w^4& w^6\\
	\bottomrule
	\end{array}
	\]
	with a dualizing vector $\mathbf{v}=( 1,1, w^6, 1, w^2, w, w^5, w^6)$.
	
	Matrices $G_2$ and $G_5$ can be realized as submatrices of $G_5$ and $G_8$ respectively. We can continue this up to a $23\times 23$ matrix according to the following lemma and corresponding functions with proper pole orders at $Q$. 
		
	\begin{lemma}
		Let $n=3i+2$ and $D_n=Q'+Q''+P_1 + \cdots + P_{3i}$. Then
		\[K+D_n \sim (n+2g-2)Q\]
		for $i=0,\ldots, 7$.
	\end{lemma}
	
	\begin{proof}
	  It was proven above for the case $i=0,1,2$.
                    From the construction of $P_i$ for $i=1, \ldots, 21$, we know that $Q'$, $P_{3i+1}$, $P_{3i+2}$, and $P_{3i+3}$ are colinear for $i=0, \ldots, 6$. Then
		\begin{align*}
		K+D_{n+3} & = K+D_n + P_{3i+1} + P_{3i+2} + P_{3i+3} \\
		& \sim (n+2g-2)Q + L - Q' \\
		& \sim (n+3+2g-2)Q
		\end{align*}
		where $L$ denotes the canonical divisor $Q'+P_{3i+1} + P_{3i+2} + P_{3i+3}$.
	\end{proof}
	Hence finding sets of points which are colinear on a line conatining $Q'$ will give a set of divisors satisfying the isometry-dual condition on the corresponding one point AG codes. For $D_{23}$ the geometric nongaps are $\{0, 3, 5, 6, 7, \ldots, 22, 23, 25, 28\}$. Then we get a sequence of isometry-dual flags 
	\[ (C_L(D_2, m_iQ))_{i=0, 1, 2} < (C_L(D_5, m_iQ))_{i=0, \ldots, 5} < (C_L(D_8, m_i Q))_{i=0, \ldots, 8} < \cdots  < (C_L(D_{23}, m_iQ))_{i=0, \ldots, 23}\]
	each of which can be obtained by properly puncturing the next flag. This agrees with Theorem~\ref{t:inh} that the difference in length $3, 6, 9, \ldots, 21$ are in $W$. Admissible pairs of the Klein curves are given by the following table.

	\[
\begin{array}{c}
\begin{array}{rccccccccccccccccccccc} \toprule
&m= &0 & &1 &2 &3 &4 &5 &6 & &7 &8 &9 &10 &11 &12 & &13 \\
& &1 & &- &- &2 &- &3 &4 & &5 &6 &7 &8 &9 &{10} & &{11} \\  \midrule
n=1 &  & \ast  & & & & & & & \\ \noalign{\medskip}
2 & & &     &  &\cdot  &\ast    &\cdot  &\ast  &\cdot  & &\ast  \\
3 & & &  & & & &\cdot  &\ast  &\ast  & &\ast  &\cdot  \\
4 & & &  & & &  & & &\cdot  & &\ast  &\ast  &\ast  \\  \noalign{\medskip}
5 & & &  & & &      & & & & &\ast  &\cdot  &\ast  &\ast  \\
6 & & &  & & &      & & & & & &   &\cdot  &\cdot  &\ast  \\
7 & & &  & & &      & & & & & &   & & &\cdot  & \ast & &  \\  \noalign{\medskip}
8 & & &  & & &      & & & & & &   & & &       & & & \ast \\ \bottomrule
\end{array} \\ \noalign{\bigskip}
\end{array}
\]

	
\end{example}

\appendix
\section{Appendices}

\subsection{Reed Muller type code} \label{s:RM}

Since the isometry-dual property of a complete flag of codes can be defined with respect to the generator matrix, we consider an example of codes not necessarily defined over a curve. Consider the affine space $\mathbb{F}_2^m$. It is an $m$ dimensional vector space over $\mathbb{F}_2$, so elements are of the form $(\alpha_m, \alpha_{m-1}, \ldots, \alpha_1)$ where each $\alpha_j \in \mathbb{F}_2 = \{0,1\}$ for $j=1, \ldots, m$. Let $x_{\alpha_j}$ be the coordinate functions for $j=1, \ldots, m$, that is, $x_{\alpha_j} ( \alpha_m, \alpha_{m-1}, \ldots, \alpha_1) = \alpha_j$ for $j=1, \ldots, m$. Note that $x_{\alpha_j}^2 = x_{\alpha_j}$. Then the coordinate ring is $R=\mathbb{F}_2 [ x_1, x_2, \ldots, x_m]/I$ for $I=(x_1^2 - x_1, x_2^2 - x_2, \ldots, x_m^2 - x_m)$, which is, as a set, a set of square free monomials in $x_1, \ldots, x_m$. Let $\alpha=(\alpha_m, \alpha_{m-1}, \ldots, \alpha_1)$ and $\beta=(\beta_m, \beta_{m-1}, \ldots, \beta_1)$ be vectors in $\mathbb{F}_2^m$. We write $x^\alpha$ for the function $x_1^{\alpha_1} x_2^{\alpha_2}\cdots x_m^{\alpha_m}$. An order on $R$ is defined by $x^\alpha < x^\beta$ if and only if
\begin{enumerate}
	\item Either $\sum \alpha_i < \sum \beta_i$
	\item or $\sum \alpha_i = \sum \beta_i$ and $\exists j$ such that $\alpha_j =0$, $\beta_j=1$ and $\alpha_k=\beta_k$ for all $k=j+1, \ldots, m$.
\end{enumerate}
Call this by {\tt DegLex } order. There is a bijection between $\mathbb{F}_2$-rational points of $\mathbb{F}_2^m$ and functions in the coordinate ring $R$ by $\alpha \longleftrightarrow x^\alpha$. We copy the {\tt DegLex } order on $R$ to $\mathbb{F}_2^m$. For a function $f=X^\alpha$ and a point $P = \beta$, 
\[ f(P) = \begin{cases}
1 \;\;\text{ if } x^\alpha | x^\beta \\
0 \;\;\text{ otherwise}
\end{cases}\]
Let $N=2^m$. Define an $N\times N$ matrix $A=(A_{f,P})$ with rows of functions in $R$ and columns of affine points in $\mathbb{F}_2^m$ both indexed by {\tt DegLex } order. Then with this orders on points and functions, we get an isometry-dual matrix. For $m=3$, we get the following matrix $A$.

\[
\begin{array}{ccccccccccccc}
\toprule 
& &000 &001 &010 &100 &011 &101 &110 &111 \\ 
\midrule
1 & &1 &1 &1 &1 &1 &1 &1 &1  \\
x_1 & &0 &1 &0 &0 &1 &1 &0 &1 \\
x_2 & &0 &0 &1 &0 &1 &0 &1 &1 \\
x_3 & &0 &0 &0 &1 &0 &1 &1 &1 \\
x_1 x_2 & &0 &0 &0 &0 &1 &0 &0 &1 \\
x_1 x_3 & &0 &0 &0 &0 &0 &1 &0 &1 \\
x_2 x_3 & &0 &0 &0 &0 &0 &0 &1 &1 \\
x_1 x_2 x_3 & &0 &0 &0 &0 &0 &0 &0 &1 \\
\bottomrule
\end{array}
\]

Then it can be easily checked that this matrix is isometry-dual with the dualizing vector $(1, 1, \ldots, 1)$, that is, we get
\[
A A^T = 
\left[ \begin{array}{cccccccc}
0 &0 &0 &0 &0 &0 &0 &1 \\
0 &0 &0 &0 &0 &0 &1 &1 \\
0 &0 &0 &0 &0 &1 &0 &1 \\
0 &0 &0 &0 &1 &0 &0 &1 \\
0 &0 &0 &1 &0 &1 &1 &1 \\
0 &0 &1 &0 &1 &0 &1 &1 \\
0 &1 &0 &0 &1 &1 &0 &1 \\
1 &1 &1 &1 &1 &1 &1 &1 \\
\end{array}
\right]
\]

For a subset of $n$ rational points, the corresponding columns in $A$ define a $N \times n$ submatrix whose row spaces define a flag of length $n$ from $0$ to $\mathbb{F}_2^n$. From this $N\times n$ matrix choose $n$ rows in a way that the chosen row is linearly independent on the previous rows. Then we get a complete flag of $\mathbb{F}_2^n$ generated by first $i$ rows of the $n\times n$ matrix. We are interested in the case when this choice gives an isometry-dual flag. The next two tables give the relevant minors that define the flag for the subsets of points $\{ 000, 001, 010, 011 \}$ 
and $\{ 000, 001, 010, 111 \}$.
\[
\begin{array}{ccccccccccccc}
\toprule 
& &000 &001 &010 &011 \\ 
\midrule
1 & &1 &1 &1 &1  \\
x_1 & &0 &1 &0 &1 \\
x_2 & &0 &0 &1 &1 \\
x_1 x_2 & &0 &0 &0 &1 \\
\bottomrule
\end{array}
\qquad
\begin{array}{ccccccccccccc}
\toprule 
& &000 &001 &010 &111 \\ 
\midrule
1 & &1 &1 &1 &1  \\
x_1 & &0 &1 &0 &1 \\
x_2 & &0 &0 &1 &1 \\
x_3 & &0 &0 &0 &1 \\
\bottomrule
\end{array}
\]
The two subsets share the same minors and thus the same flags. Both are isometry-dual. For the first subset this follows immediately with the observation
that the set is the set of all rational points in affine space of dimension $2$, that is, it is in the affine plane which is a hyperplane of a coordinate function $x_3$. The vanishing ideals for the two sets of points are
\begin{align*}
I( \{ 000, 001, 010, 011 \} ) &= (x_3), \qquad \\
I( \{ 000, 001, 010, 111 \} ) &= (x_1 x_2 + x_3, x_1 x_3 + x_3, x_2 x_3 + x_3).
\end{align*}
There are a total of 22 isometry-dual subsets of size $4$. The row span $R_5$ of the first five rows in the 8-by-8 matrix $A$, the rows labeled $1, x_1, x_2, x_3 x_1 x_2$, contains 32 vectors, with weight distribution $0\,(1\times), 2\,(4\times), 4\,(22\times), 6\,(4\times), 8\,(1\times).$ The 22 vectors of weight 4 are the characteristic vectors of the 22 isometry-dual subsets of size 4. They are divided into four groups: 2 are in $R_2 \backslash R_1$, 4 are in $R_3 \backslash R_2$, 8 are in $R_4 \backslash R_3$, and
8 are in $R_5 \backslash R_4$. Minors for subsets in the same group share the same rows.
\begin{align*}
&R_2 \backslash R_1 : 1, x_2, x_3, x_2 x_3~(2\times) \\
&R_3 \backslash R_2 : 1, x_1, x_3, x_1 x_3~(4\times) \\
&R_4 \backslash R_3 : 1, x_1, x_2, x_1 x_2~(8\times) \\
&R_5 \backslash R_4 : 1, x_1, x_2, x_3~(8\times) 
\end{align*}
Let $v=(1,1,1,0,0,0,0,1)$ be the characteristic vector for the subset $\{ 000, 001, 010, 111 \}$. 
\[
A \textnormal{diag}(v) A^T =
\left[ \begin{array}{cccccccc}
0 &0 &0 &1 &1 &1 &1 &1 \\
0 &0 &1 &1 &1 &1 &1 &1 \\
0 &1 &0 &1 &1 &1 &1 &1 \\
1 &1 &1 &1 &1 &1 &1 &1 \\
1 &1 &1 &1 &1 &1 &1 &1 \\
1 &1 &1 &1 &1 &1 &1 &1 \\
1 &1 &1 &1 &1 &1 &1 &1 \\
1 &1 &1 &1 &1 &1 &1 &1 
\end{array}
\right]
\]

The next propositions are proved computationally.

\begin{proposition}
  There are $54 = 2^6-10$ isometry-dual subsets of size 8 in the affine space $\mathbb{F}_2^4$.
  Their characteristic vectors are the vectors of weight 8 in the
  row span $R_6$ of rows corresponding to
  $1, x_1, x_2, x_3, x_4, x_1 x_2$ in $A$. The weight distribution of the row span is $0^1 4^4 8^{54} 12^4 16^1$.
The 54 subsets are divided 
into 6 orbits of sizes $2, 4, 8, 16, 8, 16.$ Pivots are in positions $(x^\alpha,x^{\alpha'})$ :
	\[
	x^\alpha x^{\alpha'} = \begin{cases}
	x_2 x_3 x_4 ~~(2\,\times) \\
	x_1 x_3 x_4 ~~(4\,\times) \\
	x_1 x_2 x_4 ~~(8\,\times) \\
	x_1 x_2 x_3 ~~(16\,\times) \\
	x_3 x_4 ~\text{or}~ x_1 x_2 x_4~~(8\,\times) \\
	x_3 x_4 ~\text{or}~ x_1 x_2 x_3~~(16\,\times) 
	\end{cases}
	\]
         
	Ideals that represent the different groups are
	\[
	I = \begin{cases}
	(x_1) \\
	(x_2) \\
	(x_3) \\
	(x_4) \\
	(x_3 + x_1 x_2,\, x_3 + x_1 x_3,\, x_3 + x_2 x_3) \\
	(x_4 + x_1 x_2,\, x_4 + x_1 x_4,\, x_4 + x_2 x_4) 
	\end{cases}
	\]

\end{proposition}

\begin{proposition}
	There are $118 = 2^7-10$ isometry-dual subsets of size 16 in affine space $\mathbb{F}_2^5$. Their characteristic vectors are the vectors of weight 16 in the
	row span $R_7$ of rows corresponding to 
        $1, x_1, x_2, x_3, x_4, x_5, x_1 x_2$ in $A$. The weight distribution of the row span is $0^1 8^4 16^{118} 24^4 32^1$.
The 118 subsets are divided 
into 8 orbits of sizes $2, 4, 8, 16, 32, 8, 16, 32.$
\end{proposition}


\begin{thebibliography}{1}

\bibitem{BdM}
Maria Bras-Amor\'{o}s and Anna de~Mier.
\newblock Representation of numerical semigroups by {D}yck paths.
\newblock {\em Semigroup Forum}, 75(3):677--682, 2007.

\bibitem{BLV}
Maria Bras-Amor\'{o}s, Kwankyu Lee, and Albert Vico-Oton.
\newblock New lower bounds on the generalized {H}amming weights of {AG} codes.
\newblock {\em IEEE Trans. Inform. Theory}, 60(10):5930--5937, 2014.

\bibitem{GMRT}
Olav Geil, Carlos Munuera, Diego Ruano, and Fernando Torres.
\newblock On the order bounds for one-point {AG} codes.
\newblock {\em Adv. Math. Commun.}, 5(3):489--504, 2011.

\bibitem{MWS}
F.~J. MacWilliams and N.~J.~A. Sloane.
\newblock {\em The theory of error-correcting codes.}
\newblock North-Holland Publishing Co., Amsterdam-New York-Oxford, 1977.
\newblock North-Holland Mathematical Library, Vol. 16.

\bibitem{MP}
Carlos Munuera and Ruud Pellikaan.
\newblock Equality of geometric {G}oppa codes and equivalence of divisors.
\newblock {\em J. Pure Appl. Algebra}, 90(3):229--252, 1993.

\end{thebibliography}

\end{document}